\documentclass[runningheads]{llncs}
\usepackage[T1]{fontenc}
\usepackage{graphicx}
\usepackage{amsmath}
\usepackage{amssymb}
\usepackage{booktabs}
\usepackage{hyperref}
\usepackage{subcaption}
\usepackage{siunitx}
\usepackage{tablefootnote}
\usepackage{comment}

\pdfoutput=1

\usepackage{orcidlink} 
\newif\ifextended
\extendedtrue   

\ifextended
  \newcommand{\extref}[2]{Appendix~\ref{#2}}
\else
  \newcommand{\extref}[2]{Appendix~#1 of the extended version~\cite{fice-extended}}
\fi

\ifextended
  \renewcommand{\orcidID}[1]{\orcidlink{#1}}
\fi

\begin{document}
\title{Fully Inductive Cardinality Estimation}
%
%
\author{%
  Tim Schwabe\inst{1}\orcidID{0009-0009-7957-603X} \and
  Lukas Ketzer\inst{1}\orcidID{0009-0008-6266-321X} \and
  Maribel Acosta\inst{1}\orcidID{0000-0002-1209-2868}%
}
 
\authorrunning{T. Schwabe et al.}
 
\institute{%
  Technical University of Munich, Munich, Germany\\
  \email{\{tim.schwabe,lukas.a.ketzer,maribel.acosta\}@tum.de}%
}
%

%
\maketitle              
\begin{abstract}
Query optimization of Basic Graph Patterns (BGP) SPARQL queries over Knowledge Graphs (KG) requires accurate cardinality estimation. Recently published learned estimators outperform statistics- and sampling-based approaches, but share a limitation preventing their adoption in real-world triplestores: they are \emph{transductive} and require retraining when the underlying graph changes or when applied to new graphs. We present \emph{FICE (Fully Inductive Cardinality Estimation)}, the first learned cardinality estimator for BGP queries over KGs that generalizes to entirely unseen graphs (including unseen relations), without any retraining. FICE is a graph neural network (GNN) with two coupled components. First, an encoder GNN over a factor-graph view of the KG produces entity and relation embeddings. We prove that BGP cardinality is a local function of the 2-hop neighborhood around bound terms in this view, motivating the local message-passing encoder. A decoder GNN then composes these embeddings along the join topology of the query to predict log-cardinality. The encoder and decoder are trained jointly, making the embeddings specialized for cardinality estimation. FICE is trained using neighborhood sampling to scale to KGs with millions of triples, and decouples embedding generation from cardinality decoding to enable estimation latency below a millisecond. Compared to learned and non-learned baselines over 10 KGs, FICE reduces the overall median q-error from 13.54 (for the best competitor) to 5.34 and dominates all approaches in tail behavior.

\end{abstract}
\section{Introduction}

Cardinality estimation is a core task in cost-based query optimization for Basic Graph Patterns in  Knowledge Graphs (KGs). It is the primary signal to guide cost models in choosing a join ordering. Inaccurate cardinality estimates can lead to query plans that run orders of magnitude slower than those guided by accurate estimates~\cite{how-good}. Further, a cardinality estimator is invoked many times inside the query optimizer's plan search, so it must be both \emph{accurate} and \emph{efficient}. While triplestores rely on statistics and heuristics to estimate cardinalities, recent work on learned cardinality estimators has shown significantly better results by using neural networks trained in a supervised way on a representative set of queries.
However, current learned cardinality estimators~\cite{gnce,lss,davitkova2021lmkglearnedmodelscardinality} share a central limitation: they are \emph{transductive}. That is, they need to be retrained when applied to a new graph or when the underlying graph changes. This prevents their adoption in real-world triplestores, since training them is usually time-consuming and KGs are typically changing continuously. What is missing are approaches that are fully \emph{inductive}, i.e., they generalize to new entities, unseen relations, as well as completely unseen graphs without any retraining. While recent work has shown inductivity on link prediction tasks~\cite{ultra}, to our knowledge, no prior cardinality estimator has been built to generalize to new KGs without retraining. Cardinality estimation differs from pure link prediction, as it requires fine-grained higher-order co-occurrence statistics of entities and relations. Further, inductive link predictors typically condition a GNN on each query over the full graph at inference time, which is fundamentally incompatible with the sub-millisecond latency budget of a query optimizer.

Building on this gap in the literature, we present \emph{FICE (Fully Inductive Cardinality Estimation)}, an encoder-decoder graph neural network for cardinality estimation that generalizes to unseen graphs without any retraining. The encoder produces entity- and relation embeddings from a factor graph view of an RDF KG without requiring entity- or relation-specific parameters. The factor graph view (i) exposes relations as first-class vertices, so a GNN can produce relation embeddings with no relation-specific embeddings, (ii) preserves structural information, and (iii) turns cardinality into a local function of the 2-hop neighborhood around bound terms in a query, which makes a local GNN sufficient to extract cardinality-relevant information. The decoder transforms a graph view of the query, together with those embeddings, into a scalar cardinality estimate. Cardinality is determined by the join topology of the query and how relations/entities co-occur in the KG. The decoder learns to compose local KG statistics (from the embeddings) along the join structure of the query. Importantly, the encoder is trained end-to-end with the decoder, so the embeddings are optimized specifically for cardinality estimation, unlike prior work that uses general-purpose pretrained embeddings. Further, during online cardinality estimation, only the decoder is needed (together with stored embeddings generated offline with the encoder; analogous to building statistics in classical estimators), enabling sub-millisecond estimation latency. 

In a leave-one-graph-out evaluation over 10 KGs, FICE achieves the best median q-error across all compared baselines, dominates in tail behavior, and generates plans with better cost compared to the production-grade triplestore optimizer \emph{QLever}.

In summary, our contributions are:
\begin{enumerate}
    \item The first fully inductive learned cardinality estimator for BGP queries, built as an encoder-decoder GNN model over a factor-graph view of the RDF KG.
    \item A 2-hop locality theorem for BGP cardinality on the factor-graph view of the KG, which motivates this graph representation and the local GNN encoder.
    \item End-to-end learning of embeddings that are specialized towards cardinality estimation.
    \item Neighborhood-sampling training scheme and offline/online decoupling to enable training on KGs with multiple million triples while retaining sub-ms inference latency.
    \item An extensive evaluation on ten KGs and diverse query shapes comparing FICE to multiple non-learned as well as learned cardinality estimators.
    \item Extending FICE with a ranking loss to further improve its accuracy in join ordering in a production-grade triplestore.
\end{enumerate}

\section{Related Work}

\paragraph{Non-Learned Cardinality Estimation}
Traditionally, cardinality estimation for query optimization has been approached using non-learned models based on statistical summaries and histograms, or sampling.
WanderJoin (WJ)~\cite{wanderjoin}, which was extended to RDF in G-CARE~\cite{gcare}, samples random walks over a join data graph view of the KG and produces unbiased cardinality estimates. It can yield accurate estimates for queries with large cardinalities, but often fails to find valid walks for selective queries.
Summary-based approaches include Characteristic Sets (CSET)~\cite{cset}, which stores per-subject predicate sets and can accurately estimate star-shaped queries but assumes independence across star-shapes. Subsequent works have shown that CSETs can be approximated to compute them faster~\cite{DBLP:journals/semweb/HelingA23,cset-sampling}. SumRDF~\cite{sumrdf} interprets a graph summary under possible-world semantics to compute expected cardinalities.
The main limitations of these approaches are that they rely on hand-crafted statistics that cannot capture all correlations important for cardinality estimation or on sampling, which struggles particularly on selective queries.

\paragraph{Learned Cardinality Estimation}
Due to the limitations of classical cardinality estimators, neural network-based approaches have been proposed in the last decade.
In relational databases, a large body of work has been proposed for learned approaches such as plan selection, hint generation, and learning-to-rank over plans~\cite{acosta2024neuro,mscn,bao,neo,learned-design-space,neurocard,lero}. PRICE~\cite{PRICE} is a fully inductive approach for relational databases that uses hand-crafted, low-level statistical features (e.g., histograms), and thus cannot be directly applied to KGs. The overall idea of pretraining across multiple datasets is similar to FICE, but our approach differs significantly in its architecture, which is designed to directly capture higher-order co-occurrence patterns from the KGs.

KGs require different approaches than relational databases due to the natural graph representation and high interconnectedness of the data~\cite{acosta2024neuro}.
LMKG \cite{davitkova2021lmkglearnedmodelscardinality} uses an MLP over a flattened representation of the query's adjacency tensor and the binary vectors for the entities and relations in the KG to regress cardinalities. It is built around path and star queries and transductive, i.e., does not learn entity and relation embeddings, so all cardinality-dependent statistics must be trained into the model itself, limiting its accuracy on new graphs.
GNCE~\cite{gnce} and LSS~\cite{lss} both represent entities and relations with pretrained embeddings. GNCE uses RDF2vec~\cite{rdf2vec}, LSS uses ProNE~\cite{prone}, and GNCE additionally adds occurrence statistics. Both approaches employ a GNN over a graph representation of the query to regress cardinalities. Lastly, SPACE~\cite{SPACE} uses an RNN to estimate the cardinalities of path queries. While GNCE is positioned as semi-inductive and theoretically supports incremental entity additions, neither of them is designed to transfer to unseen KGs without retraining, and, opposed to FICE, their embeddings are general-purpose and not trained for cardinality estimation.

\paragraph{Inductive Learning Over Graphs}
Recently, inductive approaches for the task of \emph{link prediction} in KGs have been proposed. ULTRA~\cite{ultra} is fully-inductive with a related encoder-decoder architecture to FICE. Different from our Factor Graph, it uses a relation-graph that only captures whether relations co-occur and not how often, in addition to the normal KG. Further, every prediction requires running the model over the full graph, making it incompatible with the latency constraints of online cardinality estimation.


GraphSAGE~\cite{graphsage} demonstrated that inductive node embeddings can be learned by sampling local neighborhoods, which motivates the neighborhood sampling in our encoder.

\section{FICE: Our Approach}
In this section, we present our approach, \emph{FICE (Fully Inductive Cardinality Estimation)}, the first learned cardinality estimator for KGs designed to generalize to unseen nodes, edges, and graphs without retraining. We first define the objective of FICE, and then detail the model architecture and the training procedure. We model a KG as a finite directed edge-labeled multigraph $G=(V,R,E)$, with entities $V$, relations $R$, and triples $E$. A BGP $Q$ is a finite set of triple patterns $(x,y,z)$, whose components are either variables or bound terms $(V \cup R)$. \extref{A}{app:preliminaries} provides detailed definitions.

\subsection{Optimization Objective}
We assume a distribution over knowledge graphs, and a set of training graphs $\mathcal{G}_{train}=\{G_1,\ldots, G_n\}$ sampled from this distribution. Further, each training graph has a set of training queries $\mathcal{Q}_{G_i}=\{Q_{i,1}, \ldots, Q_{i,k}\}$ associated with it. We further assume a set of test graphs $\mathcal{G}_{test}=\{G'_1,\ldots, G'_m\}$ sampled from the same distribution such that $\mathcal{G}_{train} \cap \mathcal{G}_{test} = \emptyset$ along with a set of queries for each of those graphs $\mathcal{Q}_{G'_i}=\{Q'_{i,1}, \ldots, Q'_{i,k}\}$.

Assume $l$ is a loss (e.g., MSE). Our goal is to obtain a learned cardinality estimator $\widehat{C}$ by minimizing the training loss between predicted ($\widehat{C}$) and true ($||\cdot||$) cardinality:
\begin{equation}
    \mathcal{L}_{train} = \sum_{i=1}^{n} \sum_{j=1}^{k} l(\widehat{C}(G_i,Q_{i,j}), ||Q_{i,j}||_{G_i}),
\end{equation}

with the goal of minimizing the generalization loss over all test graphs
\begin{equation}
    \mathcal{L}_{generalization} = \sum_{i=1}^{m} \sum_{j=1}^{k} l(\widehat{C}(G'_i,Q'_{i,j}),||Q'_{i,j}||_{G'_i}).
\end{equation}
That is, we aim to train our estimator on the training graphs such that it performs well on unseen graphs (with unseen entities and predicates). Compared to previous approaches, we lift the cardinality estimation problem from learning distributions of queries over a fixed graph to over a mixture of graphs.

\subsection{Model Architecture}
To achieve the generalization described above, we cannot rely on models that encode specific KG information in their parameters, nor can we use fixed embeddings for entities and relations (e.g., RDF2Vec as in GNCE), as these do not inductively generalize. Instead, FICE operates directly on the raw graph structure (and additional statistics) to transform these into a cardinality estimate of a query. We realize this by introducing two coupled GNNs, the \emph{encoder GNN} $f^E$ and the \emph{decoder GNN} $f^Q$. $f^E$ encodes the KG entities and relations into contextualized embeddings, trained specifically for cardinality estimation. $f^Q$ uses those embeddings together with occurrence statistics and a graph view of the query to decode the cardinality of the query. Figure~\ref{fig:approach} shows the FICE architecture, and the encoding and decoding phases are detailed in the following.

\subsubsection{KG Encoding Phase}

\begin{figure}[t!]
\includegraphics[width=\textwidth]{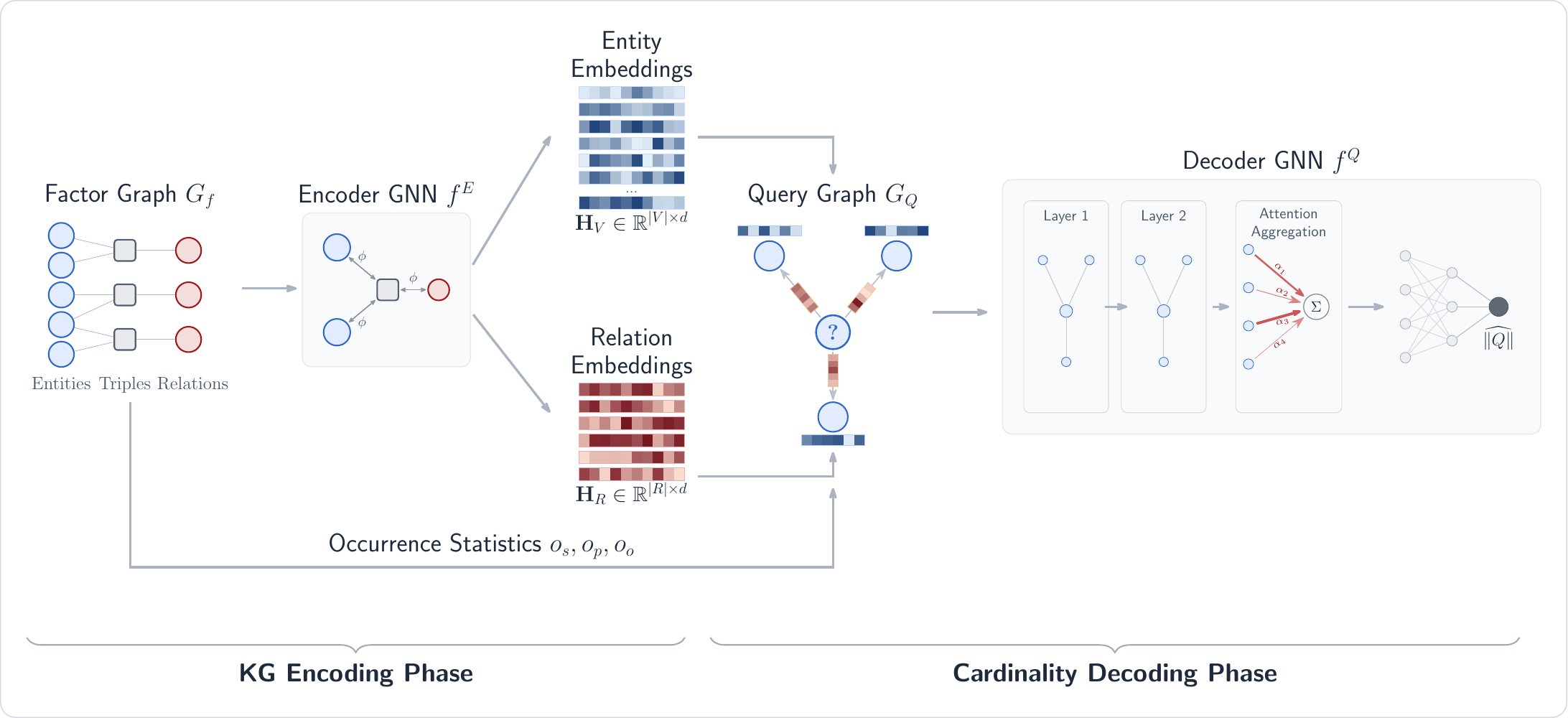}
\caption{Overview of the FICE architecture. The Encoder GNN produces embeddings for entities and relations based on a factor graph view of $G$. Importantly, the GNN consumes only the graph structure, not specific features, making it fully inductive. The Query GNN uses these embeddings alongside the query graph to produce a final cardinality estimate.} \label{fig:approach}
\end{figure}

To produce embeddings for entities and relations, we first transform $G$ into a factor graph view $G_f = (V_f, E_f)$ in which all triples $E$ and relations $R$ of the original graph are transformed into a node itself (yielding triple nodes $T$). The resulting node set is thus $V_f = V \cup R \cup T$. For each triple $(s,p,o) \in E$, we add three forward edges and their reverses to its triple node $t_{spo}$:
\begin{align}
    E_{fwd} &= \bigcup_{(s,p,o) \in E} \{ (t_{spo}, s), (t_{spo}, p), (t_{spo}, o)\} \\
    E_{rev} &= \bigcup_{(s,p,o) \in E} \{ (s, t_{spo}), (p, t_{spo}), (o, t_{spo})\}
\end{align}
and thus $E_f = E_{fwd} \cup E_{rev}$. Each edge has a role attribute $\rho$ that encodes the structural role and the edge direction:
\[
\rho(u,v) = \text{sgn}(u,v) \cdot r(u,v), \hspace{1cm} \rho \in \{-3,-2,-1,1,2,3\}\]
Here, $r$ encodes the structural role (1=subject, 2=predicate, 3=object) and sgn encodes the direction: it is positive if it is a forward edge and otherwise negative. For a KG with $|E|=|T|$ triples, $|R|$ relations and $|V|$ entities, the factor graph has $|V|+|R|+|T|$ nodes and $6|T|$ edges. 
There are several reasons to use a factor graph view. 
First, to obtain embeddings for relations using a standard message-passing GNN, the relations need to be represented as singular nodes, not scattered edges across the KG.

But most importantly, the factor graph makes the cardinality of a query a
\emph{local} function of the KG. To state this precisely, we define the
relevant neighborhood.
\begin{definition}[$n$-hop induced subgraph]\label{def:nhop}
Given the factor graph $G_f$ with nodes $V_f$ and a set $S \subseteq V_f$,
the \emph{$n$-hop induced subgraph} of $S$ is
\[
  \mathcal{N}_n(S) = G_f\Big[\{u \in V_f : \min_{b \in S} d_{G_f}(u,b) \le n\}\Big],
\]
where $d_{G_f}$ is the shortest-path distance in $G_f$.
\end{definition}
Next, note that there exists a bijection between the triple nodes in $G_f$ and the triples in $G$: $\pi: T \rightarrow E, t \mapsto (s(t), p(t), o(t))$. Defining $T_S = T \cap \mathcal{V}(\mathcal{N}_2(S))$ as the triple nodes inside the 2-hop subgraph, the following theorem holds:
\begin{theorem}[2-hop locality of cardinality in $G_f$]\label{theorem:2-hop}
Let $G=(V,R,E)$ be a finite KG with factor graph $G_f$, and let $Q$ be a BGP
with bound terms $S \subseteq V \cup R$ such that every triple pattern of $Q$
contains at least one bound term. Let $G_S := \pi(T_S)$ be the sub-KG induced
by the triples whose factor-graph nodes lie in $\mathcal{N}_2(S)$. Then
\[
  \lVert Q \rVert_G \;=\; \lVert Q \rVert_{G_S}.
\]
I.e., the 2-hop subgraph around the bound terms contains all information
needed to compute the cardinality of $Q$ exactly.
\end{theorem}
The proof is given in \extref{B}{app:proof}. The theorem implies that, given a universal approximator, the exact cardinality can in principle be computed from the 2-hop subgraph only. The factor-graph representation enables this by localizing the query's bound relations. Note that the theorem shows sufficiency of available information, not learnability by FICE. Practically, it shows that the \emph{Encoder GNN} $f^E$ does not need to process the full graph, but only a subgraph around entities/relations of interest. Thereby, it can encode in which topological constellations (and with which frequencies) a given relation cooccurs with other relations from its local neighborhood in the factor graph, which are the important statistics for cardinality estimation. We model $f^E$ as a 4-layer relational GNN with role-specific \emph{GINE}~\cite{gine,gin} message passing, where separate message functions are used for subject, predicate, and object edges, and their outputs are aggregated. Details can be found in \extref{C}{app:architecture}. Also note that we include reverse edges (together with the direction indicator $\text{sgn}$) so that messages can flow in both directions when processed with $f^E$.

For the initial features of the nodes in the factor graph, we need to choose values that are generalizable to other KGs, yet remain informative and capable of breaking symmetry between nodes (so nodes remain distinguishable). Thus, initial features are represented as $x_{v}=\big[\mathbf{z}_{base}||\mathbf{z}_{deg}\big]$. Here, $\mathbf{z}_{base} \in \mathbb{R}^1$ is a scalar with a constant value of 0.1, introduced so that the model has more dimensions to represent the entities/relations, and $\mathbf{z}_{deg} = \log(1+\text{deg}(v))$ is the logarithm of the degree of node $v$. This raw feature vector is then projected to the hidden dimension via a linear layer, followed by layer normalization and a SiLU activation.

The encoder GNN $f^E$ is applied on the factor graph supplied with those features to produce embeddings for the entities and relations (see Figure~\ref{fig:approach}). The goal of the embeddings is to contextualize their local neighborhoods and encode higher-order co-occurrence statistics (e.g., star motifs, path motifs, etc.) that are informative for the downstream task of cardinality estimation. We denote the resulting $d$-dimensional entity embeddings as $\mathbf{H}_V(G) \in \mathbb{R}^{|V| \times d}$, and correspondingly the relation embeddings as $\mathbf{H}_R(G) \in \mathbb{R}^{|R| \times d}$.

\subsubsection{Cardinality Decoding Phase}
Let $Q$ be a query and $V^b_Q \subseteq V_Q$ and $E^b_Q \subseteq E_Q$ be the bound nodes and edges in the query. Then, $\phi_V: V^b_Q \rightarrow V$ and $\phi_R: E^b_Q \rightarrow R$ map those to the corresponding entities and relations in $G$. Denote the \emph{query graph} of $Q$ as $G_Q = (V_Q, E_Q)$. Each distinct entity or variable in $Q$ becomes
a node in $V_Q$, and each triple pattern $tp=(x,y,z) \in Q$ becomes a directed edge from
$x$ to $z$ in $E_Q$, labeled with predicate (or variable) $y$.

To estimate the cardinality of a BGP query, we first featurize this query graph with embeddings produced by $f^E$ for all bound entities and relations. Nodes and edges
representing variables are represented with a zero vector of the same size as the embeddings.
To further support cardinality estimation, each bound-term embedding is concatenated with three
dimensions that encode its role-specific occurrences in $G$. Formally, for a term $v$, we define
the occurrence vector $(o_s,o_p,o_o)$ where
$o_s = |\{(s,p,o) \in G \ \text{s.t.} \  s=v\}|$, $o_p = |\{(s,p,o) \in G \ \text{s.t.} \  p=v\}|$,
and $o_o = |\{(s,p,o) \in G \ \text{s.t.} \  o=v\}|$.
The occurrence counts are log-transformed using $\log(1+\cdot)$.
These values are informative for cardinality estimation as they represent an upper bound on
the possible cardinality (via a Cartesian product). 

Using the embeddings produced by the encoder together with the occurrence statistics, we define the node and edge features for $f^Q$ as
\[
\mathbf{X}^Q_V(G,Q)[v]
=
\begin{cases}
\mathbf{H}_V(G)[\phi_V(v)] \,\Vert\, o_s(\phi_V(v)) \,\Vert\, o_p(\phi_V(v)) \,\Vert\, o_o(\phi_V(v)) & ,v\in V_Q,\\
\mathbf{0} \ \ \ ,\text{otherwise (variables)}.
\end{cases}
\]

\[
\mathbf{X}^Q_E(G,Q)[e]
=
\begin{cases}
\mathbf{H}_R(G)[\phi_R(e)] \,\Vert\, o_s(\phi_R(e)) \,\Vert\, o_p(\phi_R(e)) \,\Vert\, o_o(\phi_R(e)) & ,e\in E_Q,\\
\mathbf{0} \ \ \ ,\text{otherwise (variables)}.
\end{cases}
\]

The log-cardinality of a BGP query $Q$ is now obtained by applying $f^Q$ on $G_Q$ and the corresponding embeddings
\begin{equation}
    \widehat{C}(G,Q) = f^Q[f^E](G,Q) = f^Q \big(G_Q;\mathbf{X}^Q_V(G,Q),\mathbf{X}^Q_E(G,Q)\big).
\end{equation}
We model $f^Q$ as a 2-layer MPNN with \emph{GINE} message passing and SiLU activations. After message passing, we apply attention pooling~\cite{attention-pooling} over all node representations to obtain a single query-level vector. We use attentional pooling rather than simple sum pooling to let the model focus on the variable nodes and the most important nodes for the overall cardinality. This vector is concatenated with $\log(1+|E|)$ (the log-transformed number of triples in $G$), providing the model with a notion of graph scale. The resulting vector is fed into a two-layer MLP with SiLU activation to produce the final scalar log-cardinality estimate.

\subsection{Training and Inference}
\paragraph{Training}
The FICE model is fully differentiable. That means the loss between predicted and true log-cardinality can be backpropagated not only to $f^Q$, but also to $f^E$. Hence, the embeddings that $f^E$ produces are explicitly optimized to be predictive for cardinality estimation, opposed to previous approaches that use general-purpose embeddings~\cite{gnce}. Furthermore, by training on a sufficiently large amount of different KGs, the embeddings can become graph-agnostic and allow the model to generalize to unseen graphs (that follow a similar distribution in terms of graph topology). We train FICE using the AdamW optimizer~\cite{adamw} with a learning rate of 1e-4 and Smooth L1 loss (Huber loss) between predicted and true log-cardinality. We predict the log-cardinality rather than the raw cardinality, since the raw cardinality in realistic workloads spans many orders of magnitude, leading to training instability.

Applying $f^E$ on the full graph is both too costly (for KGs with millions of edges) and unnecessary for two reasons. First, during training, one does not need embeddings for all entities/relations, but only for those appearing in the current batch of queries. Secondly, one only needs the local graph around these according to Theorem~\ref{theorem:2-hop}. Hence, we utilize \emph{neighborhood sampling} during training, where for each batch, we sample the local 4-hop neighborhood around all target nodes occurring in the batch. We further restrict the number of nodes included in each hop to 10 (small numbers of neighbors per hop have been found to be sufficient in GraphSAGE~\cite{graphsage}). This massively reduces GPU-VRAM requirements and practically enables scaling FICE training to KGs with millions of nodes.
\paragraph{Inference}
Once the model is trained, it can be applied to unseen graphs without any further training.
Instead, one only needs to generate embeddings for all entities and relations in the graph
(using $f^E$) offline (akin to statistics generation in classical summary- and sampling 
based approaches). Crucially, this step can be decoupled from online cardinality estimation for queries.
Once the embeddings are generated, one only needs to use $f^Q$ (and a fast lookup of the stored embeddings)
to estimate query cardinalities. As we show in Section~\ref{sec:experiments}, this results in a sub-ms latency, compatible with 
the requirements of query optimizers. Updates of the underlying KG only require regenerating embeddings within two hops of the affected entities/relations, similar to statistics in classical estimators.

\section{Experimental Study}
\label{sec:experiments}
\paragraph{Datasets and Queries}
To evaluate the different cardinality estimators, we use 10 KGs from related approaches~\cite{davitkova2021lmkglearnedmodelscardinality,gnce} for cardinality estimation and link prediction, as well as from Pykeen~\cite{pykeen}. Table~\ref{tab:queries} presents the statistics of the used datasets. 
We generated a wide range of different query shapes for all datasets using an open-source tool for BGP query generation~\cite{rdf-subgraph-sampler}. For some datasets and query types, the tool has returned no results. A detailed description of the used datasets is given in \extref{E}{ap:datasets}.


\begin{table}[t]
  \centering
  \caption{Number of KG triples and queries per dataset and query type.}
  \label{tab:queries}
  \scriptsize
  \begin{tabular}{lrrrrrrrrr}
    \toprule
    \textbf{Dataset} & \textbf{Triples} & \textbf{Star} & \textbf{Path} & \textbf{Cycle} & \textbf{Flower} & \textbf{Diamond} & \textbf{Snowflake} & \textbf{Tree} & \textbf{Path+Star} \\
    \midrule
    WN18RR         & 10K &   7,249 &  20,000 &    45 & 1,000 &   116 & 1,000 & 1,000 &   --- \\
    FB15K237       & 34K &  20,000 &  32,728 &   908 & 4,958 & 5,000 & 4,738 & 4,988 & 5,000 \\
    SWDF            & 242K & 20,000 & 20,000 &   138 & 4,605 & 1,346 & 4,438 & 4,901 & 4,929 \\
    CoDEx-L         & 612K &  20,000 &  20,000 &   191 &    22 & 5,000 &   --- & 4,351 & 3,755 \\
    DBpedia100k     & 698K &  20,000 &  20,000 & 1,010 & 4,982 & 5,000 & 4,812 & 5,000 & 4,976 \\
    AIDS            & 802K &  20,000 &  20,000 &    10 & 1,860 &    28 & 2,468 &   792 &   145 \\
    Hetionet        & 2.2M &   6,059 &  20,000 &   --- &   --- &   --- &   --- &   --- &   --- \\
    LUBM            & 2.6M & 20,000 &  20,000 &   --- & 5,000 &    32 & 5,000 & 5,000 &   212 \\
    WIKIDATA        & 21M & 20,000 & 20,000 & 5,000 & 4,256 & 4,991 & 4,280 & 4,826 & 4,358 \\
    YAGO            & 58M &   5,008 &   6,000 &   --- &   --- &   --- &   --- &   --- &   --- \\
    \bottomrule
  \end{tabular}
\end{table}
\paragraph{Compared Baselines}
We compared FICE to several non-learned and learned baselines for BGP cardinality estimation. For non-learned baselines, we compared against Characteristic Sets (CSET),Wanderjoin (WJ) and SumRDF; for learned ones against GNCE, LMKG, and LSS. 


For a fair comparison, since the non-learned baselines do not train on a new graph but only build statistics/samples, we run them as intended: first build the statistics per KG, then run the cardinality estimation using these statistics per graph. For the learned approaches (including FICE), we run leave-one-out experiments to determine the inductive performance. That is, if we want to obtain the metrics for a learned approach on the $n$-th dataset, we train on all other $n-1$ datasets, and solely evaluate on the $n$-th dataset. This way, the models are guaranteed never to have seen the entities or relations in this graph (as the entities and relations between the KGs are disjoint). Since Wikidata and CoDEx-L may have overlapping relations, we do not train on one and then evaluate on the other; instead, we leave both out simultaneously. Note, however, that even if entities/relations overlap between graphs, their IRI/identity is never exposed to the approaches. FICE, GNCE, and LSS share the same high-level
encoder–decoder split: a per-graph entity/relation embedding stage
(computed by $f^E$, RDF2Vec, and ProNE respectively) followed by a
graph-agnostic decoder trained jointly on the queries of the $n{-}1$
training KGs. For LMKG, we re-index the entity and predicate
vocabularies per graph (incrementing offsets across graphs) and train a
single MLP on the union vocabulary, since LMKG has no separable encoder
stage. We view this approach to handling the compared learned approaches as maximally beneficial to those methods, and it further allows us to understand their potential for inductive cardinality estimation.

\paragraph{Implementation}
FICE is implemented in Python using PyTorch Geometric~\cite{pyg}, and we trained it on a machine with 2.3 TB of RAM and an NVIDIA H200 with 140 GB of VRAM. We used a learning rate of 1e-4 and a batch size of 32. The neighborhoods around entities and relations for the current query batch are sampled to a depth of 4 (with a maximum of 10 neighbors per node). The hidden dimension of both the decoder and encoder GNNs is 128, and the models are trained for 100 epochs.

\begin{table}[t]
  \centering
  \caption{Median q-error ($qE \downarrow$) and Pearson correlation ($r \uparrow$) between true and predicted cardinality. Best result per row in bold.}
  \label{tab:table-inductive}
  \scriptsize
  \setlength{\tabcolsep}{2.5pt}
  \begin{tabular}{l*{14}{r}}
    \toprule
      & \multicolumn{2}{c}{\textsc{FICE}}
      & \multicolumn{2}{c}{\textsc{GNCE}}
      & \multicolumn{2}{c}{\textsc{LMKG}}
      & \multicolumn{2}{c}{\textsc{LSS}}
      & \multicolumn{2}{c}{\textsc{WJ}}
      & \multicolumn{2}{c}{\textsc{CSet}}
      & \multicolumn{2}{c}{\textsc{SumRDF}} \\
    \cmidrule(lr){2-3} \cmidrule(lr){4-5} \cmidrule(lr){6-7}
    \cmidrule(lr){8-9} \cmidrule(lr){10-11} \cmidrule(lr){12-13} \cmidrule(lr){14-15}
    KG & $qE$ & $r$ & $qE$ & $r$ & $qE$ & $r$ & $qE$ & $r$ & $qE$ & $r$ & $qE$ & $r$ & $qE$ & $r$ \\
    \midrule
    WN18RR      & \textbf{3.25} & \textbf{.86} & 3.46  & .84 & 7.15   & .71 & 17.63  & .35 & 39.59         & .15          & 107.00 & $-.07$ & 16.00  & .38 \\
    FB15K237    & \textbf{5.27} & \textbf{.78} & 16.44 & .48 & 7.75   & .59 & 15.00  & .49 & 63.46         & .17          & 211.96 & .05    & 7.35   & .78 \\
    SWDF        & \textbf{2.45} & \textbf{.94} & 4.28  & .81 & 5.24   & .57 & 23.00  & .86 & 30.51         & .50          & 334.51 & $-.08$ & 14.00  & .73 \\
    CoDEx-L     & \textbf{8.19} & .74          & 10.31 & .61 & 18.66  & .56 & 23.00  & .45 & 28.08         & .29          & 359.00 & $-.01$ & 13.12  & \textbf{.77} \\
    DBpedia100k & \textbf{3.22} & \textbf{.85} & 7.56  & .74 & 6.79   & .69 & 11.00  & .52 & 41.39         & .26          & 255.00 & .00    & 123.20 & .38 \\
    AIDS        & \textbf{1.70} & \textbf{.93} & 4.00  & .92 & 78.64  & .66 & 7.00   & .80 & 26.67         & .21          & 758.76 & $-.20$ & 3.00   & .64 \\
    Hetionet    & \textbf{5.91} & \textbf{.83} & 61.86 & .70 & 304.29 & .69 & 129.78 & .50 & 36.30         & .33          & 346.56 & .05    & 463.56 & .54 \\
    LUBM        & 2.07          & \textbf{.97} & 3.06  & .93 & 8.41   & .38 & 31.13  & .86 & 9.92          & .59          & 98.83  & $-.13$ & \textbf{2.00} & .43 \\
    WIKIDATA    & 16.48         & .61          & 15.76 & \textbf{.70} & \textbf{5.99}  & .51 & 8.00 & .63 & 13.04 & .50 & 309.00 & $-.12$ & 9.00 & .66 \\
    YAGO        & 4.85          & \textbf{.86} & 8.69  & .85 & 80.19  & .48 & 14.71  & .20 & \textbf{2.92} & .70          & 197.56 & .13    & 21.42  & .50 \\
    \addlinespace
    \midrule
    Overall     & \textbf{5.34} & \textbf{.84} & 13.54 & .76 & 52.31  & .58 & 28.02  & .57 & 29.19         & .37          & 297.82 & $-.04$ & 67.27  & .58 \\
    \bottomrule
  \end{tabular}
\end{table}

\subsection{Inductive Cardinality Estimation}
\label{subsec:card-est}
We first investigate how FICE and the baseline cardinality estimators generalize to an unseen KG after being trained on the other KGs (for learned estimators). 
Table~\ref{tab:table-inductive} shows the median q-error (i.e. $\max(\widehat{C}/||Q||, ||Q||/\widehat{C})$ of the compared approaches as well as the Pearson correlation coefficient between predicted and true cardinality on all tested datasets and the macro-average across all KGs. FICE achieves the lowest overall q-error, with a nearly 2× margin over the second-best approach. Furthermore, it has the highest correlation with the true cardinality. This is desired because the predicted and true cardinalities should ideally follow a linear relationship. 

Across datasets, FICE outperforms the other approaches in both metrics. On YAGO, the largest KG, FICE beats all learned approaches by a large margin, underperforming only Wanderjoin in q-error, while still achieving the best correlation with the true cardinality. On LUBM, only SumRDF performs better than FICE, since its summaries work well on the regular synthetic structure of LUBM. Wikidata is the only graph where FICE underperforms multiple baselines. LMKG is not directly comparable, as its limitations restrict the reported metrics to star- and path-queries. These cases can be explained by FICE's training scheme and architecture. Both Wikidata and YAGO exhibit queries with much larger cardinalities than the other graphs. Since both the encoder and decoder of FICE are jointly trained on datasets that differ significantly from those cardinalities, it underperforms on them, as they entail a substantial distribution shift in graph size and query cardinality. GNCE and LSS are only partially influenced by this, since their embeddings are pretrained individually for each graph. Wanderjoin has no training step and is therefore not influenced by shifts in the distribution of graph topologies. Furthermore, because of its random walk approach, it is especially suited to denser graphs and queries with large cardinalities.  Secondly, FICE's encoder performs dense neighborhood aggregation, which can lead to oversmoothing on dense graphs. In contrast, RDF2Vec (used by GNCE) learns over singular random walks of the graph, which helps it to retain more signal for individual entities/relations in dense graphs. Both of these shortcomings of FICE can be corrected in future work, as detailed in Section~\ref{sec:conclusion}.

\begin{figure}[t!]
\includegraphics[width=\textwidth]{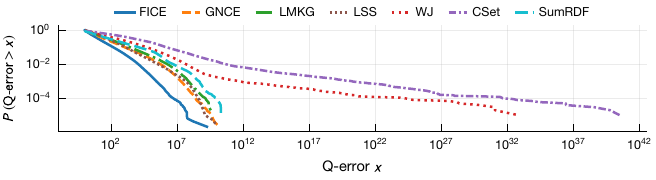}
\caption{Complementary cumulative distribution function (CCDF) of q-errors across all test queries and datasets, showing which fraction of queries have a q-error larger than $x$. FICE first-order stochastically dominates all baselines (no lines cross it) with the largest margin at the high-error tail.} \label{fig:ccdf-overall}
\end{figure}

\paragraph{Tail Behavior of Estimates}
To further understand the distributions of q-errors, Figure~\ref{fig:ccdf-overall} shows the complementary cumulative distribution function (CCDF) of q-errors, calculated over all datasets. Each point on a line shows the fraction of queries with a q-error greater than $x$. The more to the left side a line is, the better. FICE \emph{first-order stochastically dominates} all studied approaches, since no other approach crosses it to the left. That means that, for every specific q-error, strictly fewer queries in FICE have an error larger than this than in the other approaches. Figure~\ref{fig:ccdf-overall} further shows large differences in the tail behavior of the compared approaches, i.e., the rightmost end of the curves. CSET and WJ exhibit catastrophic failure modes with very large misestimates as tail behavior, while SumRDF tail behavior is close to the learned approaches. FICE exhibits the best tail behavior, orders of magnitude better than the other approaches (e.g., for q-error probabilities of 1\% and lower), and the lowest maximal q-error.
\begin{figure}[t!]
  \centering
  \begin{subfigure}[t]{\textwidth}
    \includegraphics[width=\textwidth]{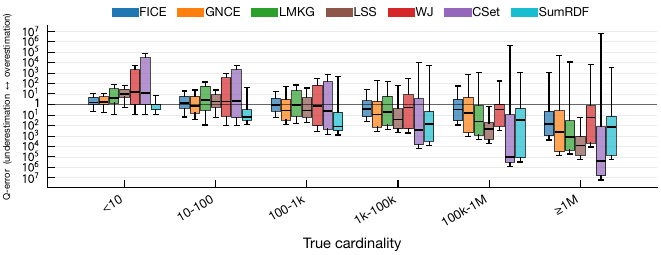}
    \caption{Q-error distribution per query cardinality}
    \label{fig:boxplot-inductive-cardinality}
  \end{subfigure}
  \\[0.6em]
  \begin{subfigure}[t]{\textwidth}
    \includegraphics[width=\textwidth]{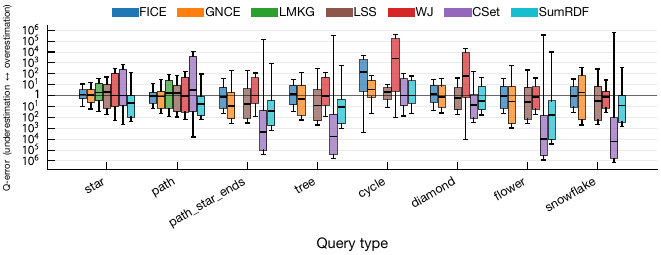}
    \caption{Q-error distribution per query shape}
    \label{fig:boxplot-inductive-query-shape}
  \end{subfigure}
  \caption{Signed q-error boxplots across all compared approaches}
  \label{fig:boxplot-inductive}
\end{figure}

\paragraph{Comparison of Cardinality Regimes}
Figure~\ref{fig:boxplot-inductive-cardinality} shows signed boxplots of the q-errors grouped by different ranges of true cardinality. The extension of the boxes above the centerline indicates overestimation by the methods, while the extension below indicates underestimation. All approaches tend to underestimate for larger true cardinalities, while being more accurate and balanced for lower ones. For most ranges, FICE exhibits the lowest interquartile ranges and (consistent with Figure~\ref{fig:ccdf-overall}) lowest tail-behavior (indicated by the whiskers). Importantly, it exhibits the lowest underestimation across all cardinality regimes. For the largest cardinalities ( $\geq$ 1M), WJ has the best median, while for small cardinalities, it behaves significantly worse than all learned approaches. Again, this is consistent with the random walk approach of WJ: for small cardinalities, it struggles to find samples satisfying the query, while for large ones it finds plenty and the estimate quickly converges to the true cardinality.

\paragraph{Comparison of Query Shapes}
Figure~\ref{fig:boxplot-inductive-query-shape} shows the results per query shape. For all shapes except cycle, flower and snowflake queries, FICE exhibits the best tail behavior, median, and interquartile range, confirming its generalizability across diverse and complex shapes. The inferior behavior compared to GNCE and LSS on cycle queries can be explained by Table~\ref{tab:queries}: there are significantly fewer cycle queries in the query dataset than in the other shapes, leading the encoder to focus on performing well on other query shapes. In contrast, the embeddings of GNCE and LSS learn independently of the query shape distribution and can thus assign equal weight to encoding cycle information as for other query shapes. WJ's behavior is heterogeneous, but also shows the worst results for cycle and diamond queries, while being equivalent to FICE on flower and snowflake queries.

\paragraph{Validating Structural Inductivity}\label{sec:struct-div}
To validate whether FICE is really inductive and generalizes to new graphs with different structure or whether the results can be explained by potential schema/structural overlap between the KGs, we correlate the per-KG median q-error with (i) the mean Jensen-Shannon divergence (JSD) between the held-out KG's structural metrics (entity-degree distribution, characteristic-set size, characteristic-set multiplicity and predicate-frequency skew) and the equal-weight mixture of them of the remaining nine KGs, (ii) log graph size, and (iii) mean JSD corrected for log graph size.

On the structural metrics, the graphs split into an open-domain cluster (FB15K237,
DBpedia100k, Wikidata, CoDEx-L, YAGO, SWDF; expected, as several derive from related ecosystems)  and a
structurally distinctive group (AIDS $0.57$, LUBM $0.38$, Hetionet $0.22$,
WN18RR $0.21$; also plausible, as these cover bio-related knowledge, synthetic data, and language taxonomies).
FICE's per-KG log median q-error correlates
\emph{negatively} with $\overline{\mathrm{JSD}}$ ($\rho=-0.65$, $p=0.04$),
also after partialling out log graph size ($\rho=-0.70$, $p=0.04$), while its
correlation with graph size alone is not significant ($\rho=0.34$, $p=0.34$). Thus, FICE does not degrade on structurally novel graphs. On the contrary, LMKG, the one baseline that explicitly depends on entity vocabulary, correlates positively with $\overline{\mathrm{JSD}}$ ($\rho=+0.40$), as expected of a transductive
method. The detailed results are presented in \extref{F}{app:divergence-correlations}.

\paragraph{Ablation Study}
To understand the effects of occurrence features and encoder GNN embeddings, we have ablated both on the DBpedia100k dataset. Removing embeddings degrades the median q-error from 3.22 to 4.55, and removing the occurrences degrades it to 3.92. Details can be found in \extref{G}{app:ablation}.\\

\noindent In summary, across KGs, cardinality regimes, and query shapes, FICE consistently yields the best median estimates and the best tail behavior for the majority of datasets. This empirically demonstrates that it can be used as an effective, fully inductive cardinality estimator across various graph sizes, query shapes, and unseen graphs without any retraining.

\subsection{Application to Join Ordering}
\begin{table}[t!]
  \centering
  \caption{Join ordering results on held-out DBpedia100k: p-error and $c$-out ratio of \textsc{FICE} vs.\ QLever's planner with Bootstrap 95\% CIs.}
  \label{tab:ranking-results}
  \footnotesize
  \setlength{\tabcolsep}{2pt}
  \begin{tabular}{l r r r r c}
    \toprule
    & \multicolumn{4}{c}{p-error}
    & \multicolumn{1}{c}{$c$-out ratio (\textsc{FICE} / \textsc{QLever})} \\
    \cmidrule(lr){2-5} \cmidrule(lr){6-6}
    Estimator & geomean (CI) & p90 & p99 & max & geomean (CI) \\
    \midrule
    \textsc{FICE} & \textbf{1.75 [1.62,\,1.89]} & \textbf{5.42}  & \textbf{195.01} & \textbf{349.72}    & \textbf{0.79 [0.71,\,0.87]} \\
    \textsc{QLever}        & 2.22 [2.00,\,2.47]         & 12.33          & 342.83          & 4{,}821.75         & 1.00 (def.) \\
    \bottomrule
  \end{tabular}
\end{table}
To gauge if FICE can improve join ordering, we compared it to the native planner of the production-grade triplestore \emph{QLever}~\cite{qlever}.

During join-ordering, closely related sub-queries that differ by a single triple-pattern must be ranked by cardinality. However, the regression loss alone optimizes absolute cardinalities (the standard objective in prior work) but does not necessarily optimize the relative order of similar subqueries. We therefore extend FICE with an auxiliary pairwise \emph{ranking loss} over sibling sub-queries. For that, we extend the training set with query sibling batches: subqueries differing only by one triple pattern. The full details can be found in \extref{H}{app:joinordering}. In our experiments, adding the ranking loss does not degrade cardinality estimation: on held-out DBpedia100k, the combined loss achieves a median q-error of $3.23$ compared to the $3.22$ of regression-loss only.

We compare this adapted version against the native planner of the production-grade triplestore \emph{QLever} in terms of 
$C_{\text{out}}$~\cite{cout} cost (the sum of all intermediate join cardinalities) on held-out DBpedia100k. For each query, we run a dynamic programming (DP) enumerator~\cite{DP,dpccp} with the estimates of FICE, and compare the resulting plan to QLever's generated plan, which also uses DP. We evaluate the $C_{\text{out}}$ ratio (FICE / QLever), i.e., the relative cost, as well as the \emph{p-error}~\cite{perror}, which is the $C_{\text{out}}$ cost ratio of the generated plan compared to the optimal plan. As Table~\ref{tab:ranking-results} shows, FICE's plans have on average a $21\%$ lower $C_{\text{out}}$ (geomean ratio $0.79$), and a p-error of $1.75$ compared to $2.22$ and, important for a production
planner, a $14\times$ smaller worst-case ($349.7$ vs.\ $4{,}822$). This shows that FICE can improve the plan estimates and prevent catastrophic tail misorderings on unseen graphs in a real triplestore.

%
\subsection{Cardinality Estimation Latency}
\begin{figure}[t!]
\includegraphics[width=\textwidth]{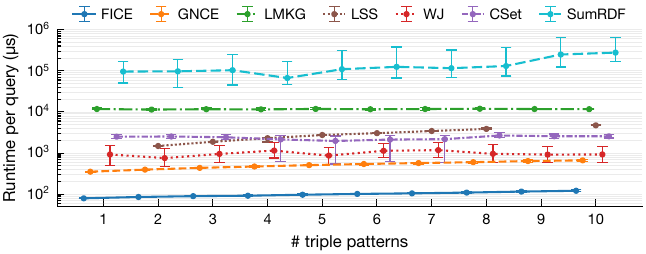}
\caption{Estimation latency (median with 25th–75th percentile error bars) of all compared approaches for increasing number of triple patterns, evaluated on all query shapes on DBpedia100k.} \label{fig:runtime}
\end{figure}
Lastly, we compare the latency of all approaches from receiving a query, loading embeddings (in case of learned approaches), transforming it into the necessary input format, running the estimation, and obtaining a cardinality estimate. For FICE and GNCE, we calculate and store the embeddings in a memory-mapped binary file and load the ones needed in the current query from it. This is a scalable approach, since large amounts of embeddings do not need to be kept in RAM. Generating FICE embeddings for all DBpedia100k entities/predicates takes 27.4s, and the resulting embeddings total 50MB. For LMKG, since no embeddings are involved, only query graph creation and model execution are measured. For WJ, we load the KG adjacency index once and, per query, measure the random-walk sampling and the final estimate based on those walks. For CSet and SumRDF, we load the precomputed summaries and then, per query, measure only the time to produce a cardinality estimate. Those measurement strategies closely resemble the latency of the methods when integrated into a real triplestore. 

Figure~\ref{fig:runtime} shows the runtime of the approaches on the DBpedia100k dataset, grouped by increasing query size, calculated over all different query shapes. FICE shows the lowest latency and negligible interquartile range, even compared to GNCE, which uses a similar GNN to map the query graph to a cardinality estimate. The reason for that is that the FICE decoder uses a simpler message-passing function than GNCE, with fewer embedding dimensions. Further, we have compiled the model and used an optimized embedding loading and query graph creation compared to GNCE. While latencies increase with query size for FICE (and all other approaches), it maintains latencies below $\sim$\qty{120}{\micro\second} even for queries with 10 triple patterns, making it practically usable as an online cardinality estimator in a production triplestore.

\section{Conclusion}
\label{sec:conclusion}
We presented \emph{FICE}, the first fully inductive learned cardinality estimator for BGP queries over RDF knowledge graphs, based on an encoder-decoder GNN architecture operating on a factor graph view of the KG. The learned embeddings are specialized for cardinality estimation, generalize to unseen graphs, and decouple the embedding generation from cardinality decoding, enabling sub-ms estimation latency. Across 10 KGs of varying sizes and densities and 8 query shapes, FICE dominates all compared approaches in terms of q-Error, exhibits the lowest estimation latency, and improves over a production-grade triplestore optimizer in terms of p-error.
\paragraph{Limitations and Future Work.}
We observed that FICE tends to underperform on Wikidata, which is significantly larger and denser than most remaining training graphs. Thus, we plan to train FICE on a larger, more diverse set of KGs. Secondly, to improve estimation on dense graphs and cyclic queries (which can be explained by an oversmoothing GNN encoder) we will add Transformer-based encoding in tandem with discriminative initial features (e.g. laplacian eigenvalues). Finally, we will integrate FICE into a production triplestore optimizer to measure its impact on query execution time.
    \paragraph*{Supplemental Material Statement:} Source Code and datasets are available online\footnote{\url{https://github.com/TimEricSchwabe/fully-inductive-cardinality-estimation}}.
\section*{Declaration of use of Generative AI}
Generative AI has been used to assist with and check the code for this submission. It has also been sparingly used to revise formulations in the manuscript of the paper.

\bibliographystyle{splncs04}
\bibliography{bibliography} 

\ifextended
\newpage
\appendix
\renewcommand{\theHsection}{appendix.\Alph{section}}
\section{Preliminaries}
\label{app:preliminaries}

\paragraph{Graph Neural Networks.} A GNN~\cite{DBLP:journals/tnn/ScarselliGTHM09,geomlearning,graph-hamilton}, produces a numerical vector,  i.e., embedding, for each node of a graph. At each of the $L$ layers, every node updates its own embedding by aggregating the embeddings of its immediate neighbors, together with edge features between them. Hence, after $L$ layers, the embedding of a node encodes its $L$-hop neighborhood. A GNN is \emph{inductive} if its parameters do not depend on the identity of specific nodes or edges and can thus be applied to unseen graphs.

\paragraph{Knowledge Graphs.}In this work, we focus on knowledge graphs modeled using RDF~\cite{rdf11-concepts}. Formally, we define a knowledge graph as follows~\cite{DBLP:conf/semweb/PerezAG06}:
\begin{definition}[Knowledge Graph]
Let $I$ be the set of IRIs, $B$ be the set of blank nodes, and $L$ be the set of literals, with $I,B,L$ pairwise disjoint. 
A \emph{Knowledge Graph} is a finite directed edge-labeled multigraph $G=(V,R,E)$, where $V \subseteq I \cup B \cup L$ is the set of entities, $R \subseteq I$ is the set of relations, and $E \subseteq (I \cup B) \times R \times V$ is the set of edges. An edge $(s,p,o) \in E$ is also called a \emph{triple} with subject $s$, predicate $p$, and object $o$.

%
\end{definition}
RDF graphs are commonly queried using SPARQL~\cite{sparql11-query}. While it offers an expressive algebra with several different operators, we focus on \emph{basic graph patterns}, as they are the focus of cost-based query optimization (specifically, join ordering). Given the set of variables $X$ such that $X \cap (I \cup B \cup L) = \emptyset$, we define a \emph{triple pattern} as a tuple $tp=(x,y,z) \in (I \cup B \cup X)\times (I \cup X)\times (I \cup B \cup L \cup X)$. The \emph{result} to \emph{matching} a triple pattern $tp$ over $G$ is the set of solution mappings $\Omega_{tp}=\{\mu: vars(tp) \rightarrow I \cup B \cup L \ | \mu(tp) \in G\}$. I.e., a mapping is a function mapping from the variables in $tp$ to $I \cup B \cup L $.

A basic graph pattern $Q$ is a conjunction of several triple patterns, i.e. $Q=\{tp_1,\ldots,tp_n\} \subseteq (I \cup B \cup L \cup X)^3$. A mapping  $\mu : \mathrm{vars}(Q) \to (I \cup B \cup L)$ induces the graph $\mu(Q) := \{\mu(tp) \mid tp \in Q\}$. Here, $\mu(tp)$ denotes the triple obtained by replacing each variable $v$ in $tp$ by $\mu(v)$. We call $\mu$ a \emph{solution} of $Q$ over $G$ iff $\mu(Q) \subseteq G$. Then, the \emph{query result} is the set $\Omega_Q := \{\mu \mid \mu(Q)\subseteq G \ \wedge\ \mathrm{dom}(\mu)=\mathrm{vars}(Q)\}.$
Importantly, we can then define the cardinality of a query as
\begin{definition}[Query Cardinality]
The cardinality of $Q$ over $G$ is the number of distinct solutions:
\[
    \lVert Q \rVert_G := |\Omega_Q|.
\]
\end{definition}

\section{Proof of Theorem 1}
\label{app:proof}

\begin{figure}[h!]
\includegraphics[width=\textwidth]{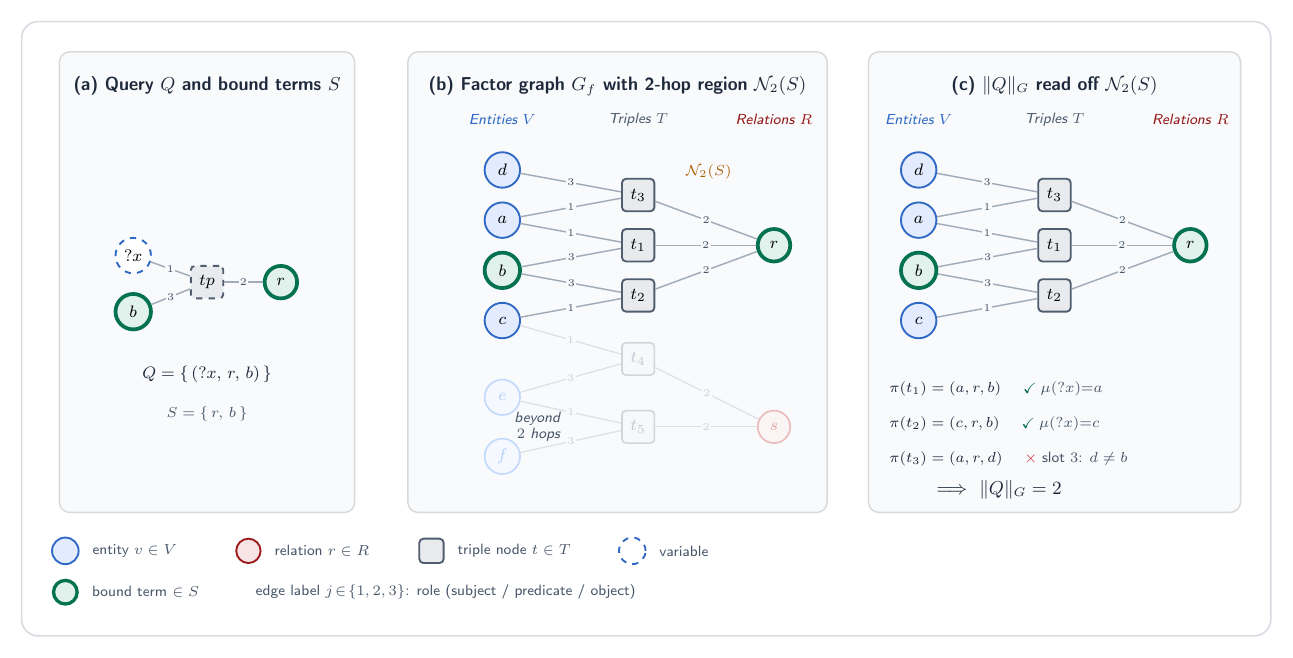}
\caption{Illustration for the proof on locality. All triple nodes and hence subjects, predicates, and objects that can potentially fulfill a BGP query $Q$ are within 2 hops of the bounded terms $S$ appearing in the query within $G_f$.} \label{fig:locality}
\end{figure}

\paragraph{Intuition}
The proof proceeds in two stages: First, we show that each per triple-pattern result set $\Omega_i(G)$ depends on $G_f$ only through the 2-hop induced subgraph around the bound terms in the triple pattern. And secondly, since the overall BGP solution set is the join $\Omega(Q,G) = \Omega_1(G) \bowtie \cdots \bowtie \Omega_k(G)$, and the join is a deterministic function of its input sets, the join itself, and its cardinality are just a function of the 2-hop subgraph. The theorem thus motivates the use of a local GNN with a limited receptive field (two suffice in theory; we use four in $f^E$ so a relation/entity node can see others directly) to estimate the cardinality. The intuition of the FICE model is that the encoder can encode/approximate the subgraph around the bound nodes, while the decoder GNN approximates the join between encoded neighborhoods. Note that the theorem does not guarantee learnability of this function by the encoder-decoder model, but instead motivates why a local message-passing GNN is used to encode embeddings. Also note that while the proof restricts to triple patterns with at least one bound term, FICE is readily applicable to BGPs with triple patterns with only variables.

Figure~\ref{fig:locality} illustrates the rationale of the proof. Panel (a) shows a single triple pattern $(?x,r,b)$ with bounded predicate and object. Panel (b) shows the 2-hop induced subgraph around $b$ and $r$, which includes all triple nodes whose triples could potentially answer the query. Panel (c) shows that the cardinality can be computed from this subgraph by projecting the triple nodes and its adjacent neighbors to triples via $\pi$, checking against the triple pattern and counting all valid ones.

\paragraph{Preliminaries}
Given the factor graph $G_f$ with nodes $V_f$ and a set of nodes $S \subseteq V_f$, the \emph{n-hop induced subgraph} of $S$ in $G$ is
\begin{equation}
    \mathcal{N}_n(S) = G_f\Big[ \{u \in V_f: \min_{b \in S} d_{G_f}(u,b) \leq n\}\Big].
\end{equation}
Here, $d_{G_f}$ is the shortest path distance in $G_f$. Note that $\mathcal{N}_n(S)$ is a role-labeled multigraph and preserves the parallel edges between vertices (in cases of self-loop triples). Next, note that there exists a bijection between the triple nodes in $G_f$ and the triples in $G$: $\pi: T \rightarrow E, t \mapsto (s(t), p(t), o(t))$

Following standard SPARQL notation, we define the solution set of a triple pattern $tp_i$ as $\Omega_i = \{\mu: \text{dom}(\mu) = \text{vars}(tp_i), \mu(tp_i) \in E\}$ and the solution set of a BGP with $k$ triple pattern as $\Omega(Q, G)= \Omega_1(G) \bowtie \ldots \bowtie \Omega_k(G)$, where $\bowtie$ is the join. For clarification, the bound terms in $Q$ are the non-variable terms. We also extend the definition of solution mappings $\mu$ to bound terms via the identity, i.e. $\mu(b)=b, \forall b \in (I \cup L)$. We now proceed with the proof of the theorem.
\begin{proof}
    Given the triple nodes $T$ in $G_f$, define $T_S = T \cap \mathcal{V}(\mathcal{N}_2(S))$, where $\mathcal{V}(\mathcal{N}_2(S))$ is the vertex set of the 2-hop induced subgraph. We show that for each triple pattern $tp_i \in Q$,
\begin{equation}
    \label{eq:set-equivalence}
    \Omega_i(G) = \{\mu: \text{dom}(\mu)=\text{vars}(tp_i), \mu(tp_i)=\pi(t) \text{ for some } t \in T_S\},
\end{equation}
by showing inclusion in both directions:\\
$(\subseteq):$ Let $\mu \in \Omega_i(G)$ and set $t=\pi^{-1}(\mu(tp_i)) \in T$, which exists since $\pi$ is a bijection. By assumption, there exists a bound term in $tp_i$. Assume the bound term $b$ is in position $j$.  Since $\mu$ is the identity on bound terms, the $j$-th component of $\pi(t)=\mu(tp_i)$ equals $b$. Hence, $t$ is adjacent to $b$ in $G_f$ and thus $d_{G_f}(t,S)=1$ and $t\in T_S$.\\
$(\supseteq):$ By construction, $\mu(tp_i)=\pi(t)$ and $\text{dom}(\mu)=\text{vars}(tp_i)$ for some $t$, and hence $\mu(tp_i)=\pi(t) \in \pi(T) =E$, so $\mu \in \Omega_i(G)$.

It remains to show that the right-hand side of~\eqref{eq:set-equivalence} depends on $G$ only through $\mathcal{N}_2(S)$. Since $G_f$ is bipartite with parts $T$ and $V \cup R$, and for all bound terms it holds that $S \subseteq V \cup R$, so $d_{G_f}(t, S)$ is odd for every $t \in T$. Since also $d_{G_f}(t,S) \leq 2$, it follows that $d_{G_f}(t,S)=1$. The three role-labeled neighbors $s(t),p(t),o(t)$ are then at a distance $\leq 2$ from S and lie in $\mathcal{V}(\mathcal{N}_2(S))$, with their role-edges incident to $t$. Hence, $\pi(t)=s(t), p(t), o(t)$ can be reconstructed from $\mathcal{N}_2(S)$ for every $t \in T_S$ and thus \eqref{eq:set-equivalence} shows that $\Omega_i(G)$ is a function of $\mathcal{N}_2(S)$ alone.

Finally, since the full BGP solution set $\Omega(Q, G)=\Omega_1(G) \bowtie \ldots \bowtie \Omega_k(G)$ is a deterministic function of only the per-pattern result set, its cardinality $\lVert Q\rVert_G=\lvert\Omega(Q, G)\rvert$ depends on $G$ only through $\mathcal{N}_2(S)$, i.e. $||Q||_G=||Q||_{G_S}$.
\qed
\end{proof}

 \section{Model and Training Details}
 \label{app:architecture}

 Here, we detail the model architecture of both the encoder and decoder GNN and provide further details on the training scheme that the main text does not cover. A detailed illustration of the encoder and decoder GNNs is given in figures~\ref {fig:encoder} and~\ref {fig:decoder}.

 \subsection{KG Encoder $f^E$}
 The encoder GNN $f^E$ produces a $d=128$-dimensional
 embedding for every node (i.e., entity, relation, or triple). However, only the entity and
 relation embeddings are subsequently utilized by the decoder.

 \paragraph{Input features.}
 For every node $v \in G_f$, we build a 3-dimensional raw feature vector
 $[z_{\text{base}}\,\|\,z_{\text{deg}}\,\|\,z_{\text{seed}}]$, where
 $z_{\text{base}}=0.1$ is a constant,
 $z_{\text{deg}}=\log(1+\deg(v))$ is the scalar log-degree, and
 $z_{\text{seed}}=0$ is used throughout. The seed channel is
 retained only to keep the input projection shape-compatible with the option to condition
 the embeddings on the query (which is not used in this paper).
 The raw vector is projected to $d=128$ by a linear
 layer, followed by layer-norm and SiLU activation with a dropout rate of $p=0.1$.

 Next, a learned \emph{node-type} embedding $\mathbb{R}^{3\times d}$ is added to the
 projected features to distinguish entity, relation, and triple nodes
 (indices $0,1,2$ respectively). Note that those are graph-agnostic since they only 
 encode node roles that are present in any KG.

 The directional edge attributes $\rho$ are projected to $\mathbb{R}^d$ by a
 single linear layer and are provided as edge features to every message-passing
 layer.

 \paragraph{Message passing.}
 The role of message-passing is to propagate the node's current information/features
 to its neighbors, so they can encode what is around them. Since the messages that a
 triple node sends to a subject might be very different than what it should send to
 a relation node, we are using a relational message-passing layer that learns 
 separate parameters for each type of role (subject, predicate, or object connection).
 Concretely, $f^E$ consists of $L_E=4$ stacked \emph{relational GINE} layers.
 Layer $\ell$ takes node features
 $\mathbf{h}^{(\ell-1)}\in\mathbb{R}^{N_f\times d}$, edge indices, edge
 features $\mathbf{e}$, and the \emph{absolute} role $|\rho|\in\{1,2,3\}$ for each edge.
 For each role $r$, a role-specific message MLP
 \[
 \phi_r(\mathbf{h}_u,\mathbf{e}_{uv}) = \sigma\!\bigl(W_r\,[\mathbf{h}_u\,\|\,\mathbf{e}_{uv}]\bigr),\quad
 \sigma=\text{SiLU},
 \]
 computes the individual messages, which are then summed per destination node.
 The final node update applies a shared MLP to the aggregated message:
 \[
 \mathbf{h}^{(\ell)}_v = \text{MLP}_{\text{upd}}\Bigl(
 (1+\varepsilon)\,\mathbf{h}^{(\ell-1)}_v + \sum_{r\in\{1,2,3\}}\sum_{u:(u,v)\in E_f^r}\phi_r(\mathbf{h}_u,\mathbf{e}_{uv})
 \Bigr),
 \]
 where $\text{MLP}_{\text{upd}}$ is a 2-layer MLP with SiLU activation function and a dropout rate of $0.1$, and $\varepsilon$ is a learnable scalar initialized to $0$. After every message-passing layer, we apply layer normalization
 and a residual connection for stability, as well as a further dropout step with rate $0.1$.

The node features after the last layer yield the final entity and relation
 embeddings $\mathbf{H}_V(G)\in\mathbb{R}^{|V|\times d}$ and
 $\mathbf{H}_R(G)\in\mathbb{R}^{|R|\times d}$.

 \subsection{Query Decoder $f^Q$}
 The decoder $f^Q$ is a 2-layer GINE~\cite{gine}-based GNN operating on the query graph $G_Q$
 with node/edge features. The per-node feature dimension
 is $d_Q = d+3 = 131$ (128-dim encoder embedding concatenated with the
 log-transformed occurrence triple $[o_s,o_p,o_o]$. As stated, variable feature dimensions are all set to zero. Edge features have the same $131$ dimensions, since again the 128-dimensional embeddings for
 relations as well as the log-occurrences of the relations are used.

 \paragraph{Message passing.}
We use a standard GINE message passing in the decoder. Unlike $f^E$, we use a single MLP that operates on the sum of all neighbor feature vectors for a given node. Hence, node features update like
 \[
 \mathbf{h}^{(\ell)}_v = \text{MLP}_{\text{upd}}\Bigl(
 (1+\varepsilon)\,\mathbf{h}^{(\ell-1)}_v + \sum_{u \in \mathcal{N}_v}\text{ReLU}(\mathbf{h}_u +\mathbf{e}_{uv})
 \Bigr),
 \]
 The Update-MLP once more has 2 layers with SiLU activation.

 The message-passing layers map the node features from $131\!\to\!131$ for the first layer and $131\!\to\!d_R=200$
 for the second (read-out) layer. After each layer, we apply SiLU and
 layer normalization. No skip-connections are used for the decoder, as we found that this hurts the model performance.

 \paragraph{Readout.}
 We use attentional pooling~\cite{attention-pooling} with a gate
 network, yielding a single $\mathbb{R}^{200}$ vector $\mathbf{r}$ representing the query. Concretely, the attentional pooling collapses all individual node feature $\mathbf{h}^{(2)}$ vectors like
 \[
 \mathbf{r} = \sum_{n=1}^{N} \mathrm{softmax} \left(
        h_{\mathrm{gate}} ( \mathbf{h}_n^{(2)} ) \right) \cdot
        \mathbf{h}_n^{(2)} ,
 \]
 The gate network is another small 2-layer MLP with SiLU activation.

 \paragraph{Head.}
 The pooled vector is concatenated with the log-graph size $\log(1+|E(G)|)$
 and passed through a final 2-layer MLP($\text{Lin}_{201\to 50}\!\to\!\text{SiLU}\!\to\!\text{Lin}_{50\to 1}$). The 
absolute value of the scalar output of this network represents the log-cardinality estimate.

 \subsection{Training}
 \paragraph{Loss.}
 We minimize the Smooth L1 (Huber) loss with $\beta=1$ between the predicted
 and ground-truth log-cardinalities, where the ground-truth is
 $\log(1+\lVert Q\rVert_G)$ (we use $\log1p$ to keep empty-result queries
 numerically stable).

 \paragraph{Optimizer.}
 AdamW ($\beta_1=0.9$, $\beta_2=0.999$) with learning rate $10^{-4}$ and
 weight decay $0$. All gradients are clipped to a global norm of $1.0$ for training
 stability.

 \paragraph{Batching and sampling.}
 Training uses a batch size of $32$ queries. As stated in the main body, applying $f^E$ on the
 full factor graph is intractable for multi-million-edge KGs, which is why we run
 neighborhood sampling with fan-out $[10,10,10,10]$ around the set of
 entities and relations referenced by the current query batch, matching the
 number of message-passing layers $L_E=4$. Query batches are drawn from a
 multi-graph loader that weights each (dataset, query-set) pair by its size,
 so that, in every epoch, the model sees queries from all training KGs in proportion to their population sizes.
 This prevents overfitting to particular graphs that might be overrepresented in the query datasets.

 \paragraph{Training schedule.}
 Models are trained for up to $100$ epochs on a single NVIDIA H200 (140 GB VRAM). A single epoch
 takes approximately 10 minutes on all graphs and query datasets.

\section{Baseline Implementation Details}
We have used the official code provided by the baseline approaches, except for the non-learned approaches. For Wanderjoin, SumRDF, and Characteristic Sets, we use the implementations provided by G-Care\footnote{\url{https://github.com/yspark-dblab/gcare}}. We use the respective linked code repositories for GNCE\footnote{\url{https://github.com/DE-TUM/GNCE}}, LMKG\footnote{\url{https://github.com/DamjanGjurovski/LMKG-Code}}, and LSS\footnote{\url{https://github.com/Kangfei/LSS}}.
As explained, for GNCE, we generate RDF2Vec-embeddings separately for all graphs, and then train jointly on the queries of the $n-1$ training graphs and evaluate on the held-out graph. To enable multi-graph training for LSS, all graphs (including the held-out graph) are combined into a large single graph (which thus has disconnected subgraphs corresponding to the used KGs). Then, we train only on queries corresponding to the training graphs and evaluate by running them over the held-out graph. For LMKG, we generate one large graph by incrementing the ids of entities and relations from graph to graph, and similarly train on the $n-1$ graphs and evaluate on the queries of the held-out graph.

\section{Details on the used datasets}
\label{ap:datasets}
WN18RR is a filtered subset from WordNet and represents a lexical-semantic graph of English synsets with semantic and lexical relations. It is thus a rather homogeneous and structured KG. FB15K237 is based on a subset of Freebase and models broad real-world knowledge across entities and relation types, making it more heterogeneous than WN18RR. SWDF is a scholarly knowledge graph describing papers, people, organizations, and conferences. It is domain-specific and relatively structured/homogeneous. DBpedia100k is a subset of DBpedia that represents broad cross-domain knowledge in a heterogeneous structure. LUBM is a synthetic benchmark generated from a university-domain ontology and is the most homogeneous among the datasets used. Wikidata (subset) is a large, multilingual KG about world knowledge, and thus the most heterogeneous among the tested ones. CoDEx-L is the largest variant of the CoDEx benchmark and is based on Wikidata. AIDS is derived from the NCI AIDS antiviral screening dataset and consists of many small molecular graphs merged into a single RDF graph. Hetionet is a biomedical knowledge graph integrating genes, diseases, compounds, anatomy, and biological processes from multiple databases; it is domain-specific and heterogeneous. Finally, YAGO is the largest KG in our evaluation and covers general-purpose knowledge from diverse sources.

\section{Structural Divergence Analysis}
\label{app:divergence-correlations}

Next, we detail the structural-divergence analysis performed in section~\ref{sec:struct-div}. We want to determine whether the cross-KG transfer of FICE is either just mediated by schema similarity/overlap between the training and held-out graphs, or whether FICE indeed generalizes to structurally novel KGs. Note that any direct IRI overlap between KGs is impossible by construction, as IRIs are never forwarded to the FICE model (the entities/relations are solely encoded by their degree and structural neighborhood in $G_f$). 
\subsection{Structural Metrics}
To do this, we compare the KGs based on four different \emph{structural metrics}:
\begin{description}
\item[Degree (DEG)] The distribution of total entity degree
  $\deg(e) = |\{(e,\cdot,\cdot)\} \cup \{(\cdot,\cdot,e)\}|$, which describes the hubs of the KG.
\item[Characteristic-set size (CSIZE)] For each subject entity $e$,
  the number of distinct predicates $|\mathrm{CS}(e)| = |\{p \,:\, (e,p,\cdot) \in G\}|$, which captures how rich and diverse the entities adjacent edges are on average.
\item[Characteristic-set multiplicity (CSET\_MULT)] For each
  distinct characteristic set, the number of entities that have it.
  This captures the regularity of the schema: e.g., synthetic KGs typically
  have few distinct $\mathrm{CS}$ values and each is shared
  by many entities, while heterogeneous real-world KGs have many rare $\mathrm{CS}$ values.
\item[Predicate-frequency skew (PRED\_FREQ)] For each predicate $p$,
  the normalized occurrence frequency $\mathrm{freq}(p)/|G|$. The resulting distribution
  captures how power-law or uniform predicate usage is.
\end{description}
To summarize the training mixture per held-out KG, we average the fingerprint distributions of all training graphs. Then, we calculate the Jensen-Shannon divergence (jsd) between the distribution of the held-out graph and the training mixture for each structural metric, and average them to obtain a scalar $\overline{\mathrm{JSD}}(k)$.

Table~\ref{tab:divergence-raw} reports the per metric JSDs and their mean for each of the ten KGs. Two groups emerge: a cluster
of open-domain KGs (fb237, dbpedia100k, wikidata, codex\_l, yago, swdf) with
$\overline{\mathrm{JSD}} \le 0.2$, and a set of structurally distinctive KGs, namely aids, lubm, hetionet, and wn18rr, with notably higher divergences. This means that those graphs individually are very different from the combination of the rest of the graphs, while the graphs in the first cluster are all very similar to the combinations of the others. This means that if an approach performs well on one of the divergent graphs, it supports the claim that the approach is genuinely inductive and can generalize to structurally novel graphs.

\begin{table}[h!]
\centering
\caption{Jensen--Shannon divergence between each held-out KG and the mixture of the remaining nine training graphs, per structural fingerprint.
A higher value means more structurally novel relative to the
training mixture.}
\label{tab:divergence-raw}
\footnotesize
\begin{tabular}{lccccc}
\toprule
KG & \textsc{deg} & \textsc{csize} & \textsc{cset\_mult} & \textsc{pred\_freq} & Mean \\
\midrule
swdf        & 0.18 & 0.29 & 0.12 & 0.09 & 0.17 \\
aids        & 0.36 & 0.47 & 0.87 & 0.57 & 0.57 \\
lubm        & 0.05 & 0.33 & 0.85 & 0.31 & 0.38 \\
fb237       & 0.05 & 0.01 & 0.23 & 0.11 & 0.10 \\
codex\_l    & 0.30 & 0.23 & 0.15 & 0.08 & 0.19 \\
dbpedia100k & 0.06 & 0.02 & 0.16 & 0.12 & 0.09 \\
wikidata    & 0.06 & 0.01 & 0.16 & 0.14 & 0.09 \\
yago        & 0.16 & 0.33 & 0.09 & 0.14 & 0.18 \\
hetionet    & 0.28 & 0.10 & 0.31 & 0.19 & 0.22 \\
wn18rr      & 0.08 & 0.34 & 0.19 & 0.25 & 0.21 \\
\bottomrule
\end{tabular}
\end{table}

\subsection{Correlations with per-KG q-error}

Table~\ref{tab:divergence-correlations} reports Pearson correlations of each
method's log median q-error against (a) mean JSD, (b) $\log |G|$, and (c)
mean JSD after partialling out $\log |G|$. The partial correlation to factor out the influence on the graph size is calculated as ($X$= mean JSD, $Y$=median q-error, $Z$=graph size).

\[
\rho_{XY \mid Z} \;=\; \frac{\rho_{XY} - \rho_{XZ}\,\rho_{YZ}}{\sqrt{(1-\rho_{XZ}^2)(1-\rho_{YZ}^2)}},
\]

The results support what is theoretically expected: WanderJoin is negatively correlated with graph size ($\rho=-0.81, p=0.008$); its sampling gets more accurate with larger graphs, and its q-error correlation to JSD is small and nonsignificant ($\rho=-0.13, p=0.72$); there is no inductive learning step. Similarly, CSet shows no significant correlation(JSD:$\rho=0.26,p=0.47$, size:$\rho=0.19,p=0.60$), as its statistics are recomputed per graph. For LMKG, the only learned baseline whose parameters are vocabulary-specific (no meaningful embeddings), the q-error correlates with JSD ($\rho=0.4,p=0.25$), which is expected from a transductive method (though non-significant). GNCE's correlation with JSD is non-significant ($\rho=-0.44,p=0.21$), while LSS has no correlation ($\rho=-0.06, p=0.88$). Lastly, FICE shows a statistically significant negative correlation with JSD ($\rho=-0.65,p=0.04$), which is sustained even after correcting for graph size ($\rho=-0.7,p=0.03$) and a nonsignificant correlation with graph size($\rho=0.34,p=0.34$). It needs to be added that those results are obtained with a small sample size of only 10 graphs, and thus need to be interpreted cautiously.

\begin{table}[t]
\centering
\caption{Pearson correlation of $\log$ median q-error with mean JSD,
$\log |G|$, and partial JSD $\mid \log |G|$.}
\label{tab:divergence-correlations}
\footnotesize
\begin{tabular}{lccc}
\toprule
Method & $\rho(\log\mathrm{qE},\,\overline{\mathrm{JSD}})$
       & $\rho(\log\mathrm{qE},\,\log |G|)$
       & $\rho(\log\mathrm{qE},\,\overline{\mathrm{JSD}} \mid \log |G|)$ \\
\midrule
FICE  & \textbf{$-0.65\ (0.044)$} & $+0.34\ (0.34)$          & \textbf{$-0.70\ (0.036)$} \\
GNCE  & $-0.44\ (0.21)$           & $+0.26\ (0.47)$          & $-0.46\ (0.21)$          \\
LMKG  & $+0.40\ (0.25)$           & $+0.39\ (0.27)$          & $+0.42\ (0.26)$          \\
LSS  & $-0.06\ (0.88)$           & $-0.02\ (0.95)$          & $-0.06\ (0.89)$          \\
WJ    & $-0.13\ (0.72)$           & \textbf{$-0.81\ (0.005)$} & $-0.18\ (0.65)$          \\
CSet  & $+0.26\ (0.47)$           & $+0.19\ (0.60)$          & $+0.26\ (0.50)$          \\
SumRDF  & $-0.47\ (0.17)$           & $+0.07\ (0.81)$          & $-0.47\ (0.20)$          \\

\bottomrule
\end{tabular}
\end{table}

\section{Ablation Study}
\label{app:ablation}
Table~\ref{tab:ablation} shows the effect of either removing the embeddings from the encoder GNN or the occurrence features from FICE on the held-out DBpedia100k dataset. When ablating the embeddings, all embedding dimensions except for the three count dimensions are set to zero. Hence, the decoder GNN can only use the query graph structure and the occurrences of the entities/relations as information to predict cardinality. When ablating for occurrences, the three corresponding dimensions are set to zero for all bound entity/relation embeddings. The results show that removing either significantly hurts the performance of FICE. Removing the occurrences, however, hurts less, and the correlation between true and predicted cardinality is retained. This is expected, since the embeddings hold important co-occurrence information that greatly influences the selectivity and thus cardinality of joins in the query. Further, absolute counts can at least be approximated from the sampled neighborhood the embeddings encode (as per Theorem 1 in the main paper). Together, the ablation results motivate why we use both the learned embeddings via the encoder GNN and the cheap occurrence features in FICE.

\begin{table}[t]
\centering
\caption{Ablation Study of FICE with dbpedia100k as the held-out dataset. Either removing the embeddings from the encoder GNN or the occurrence features hurt the model.}
\label{tab:ablation}
\begin{tabular}{lcc}
\toprule
Method & Median $qE$ & log-Pearson $r$ \\
\midrule
FICE & \textbf{3.22} & \textbf{0.85} \\
FICE w/o embeddings & 4.55 & 0.76 \\
FICE w/o occurrences & 3.92 & 0.84 \\
\bottomrule
\end{tabular}
\end{table}

\section{Application to Join Ordering}
\label{app:joinordering}
    The cardinality regression loss of FICE alone does not suffice to serve as a good estimator to aid in join ordering. The reason is that join ordering does not depend on absolute cardinalities but on the relative ordering of adjacent sub-queries, as shown in relational databases by Negi et al.~\cite{flowloss}. During ordering, the optimizer needs to compare very related sub-queries that differ by a single triple pattern. The cardinality differences between those are often small, and a pure regression loss drives the model to get cardinality predictions into the right order of magnitude, but not necessarily to differentiate similar sub-queries. To address this, we add an auxiliary \emph{pairwise ranking loss} over sub-queries of a query, related in spirit to what \emph{Lero} is doing in relational databases~\cite{lero}.

\subsection{Pairwise Ranking Loss with Size and Top Weighting}
To apply a ranking loss, we need to extend the existing query datasets to have \emph{sibling} subqueries that can be directly compared. For that, we sample, for each parent BGP, a chain of sub-BGPs of increasing size:
the chain starts at a random single triple and at every step extends by
adding one connected triple, then adding all candidate extensions as
\emph{siblings} of one another up to a maximum number of 15 pairs per depth. This produces a collection of sibling
pairs at every depth of the join tree. Each pair shares the same child
and differs only in the next triple. Ranking these siblings accurately is exactly the target
that FICE needs to fulfill to be beneficial in join ordering.

Let $\hat y_i$ be the model's predicted log-cardinality of sub-BGP $i$,
and $y_i$ its true log-cardinality. For a sibling pair $(i,j)$ with
$y_i < y_j$ (i.e.\ $i$ is the cheaper join), we use a RankNet-style~\cite{RankNet}
binary cross-entropy on the score difference,
\[
\ell_{ij} \;=\; \mathrm{BCE}_{\mathrm{logits}}\!\bigl(\hat y_j - \hat y_i,\,1\bigr)
\;=\; \log\!\bigl(1 + e^{-(\hat y_j - \hat y_i)}\bigr).
\]
This smoothly penalizes wrong rankings and fades to zero once the
predicted order matches the true order.

Additionally, we weight pairs by two multiplicative terms. Let $|s_i|$ denote
the size (number of triples) of sub-BGP $i$. The \emph{size weight}
\[
w^{\mathrm{size}}_{ij} \;=\; \bigl(|s_i|\cdot|s_j|\bigr)^{\alpha/2}
\]
rebalances the loss toward deeper sub-BGPs. The reason for this is that the
above procedure samples more pairs
at shallow depths than deep ones. However, the more difficult choices are typically
deeper within the join tree, and the goal is to focus FICE on those.
Further, we want FICE to focus on ordering the siblings at and near the argmin of
the subquery cardinalities. This is arguably more important than ranking the worst siblings
correctly. Let $r_i$ denote
the rank of $i$ within its sibling group sorted by ascending true
cardinality (so $r_i = 0$ is the argmin). The \emph{top weight}
\[
w^{\mathrm{top}}_{ij} \;=\; \frac{1}{1 + \beta\,\min(r_i, r_j)}
\]
emphasizes pairs near the argmin.

Together, the ranking loss, summed over all pairs $\mathcal{S}$ (that are siblings) in a
batch and normalized by the above weights, is
\[
\mathcal{L}_{\mathrm{rank}} \;=\;
\frac{\sum_{(i,j)\in\mathcal{S}}
       w^{\mathrm{size}}_{ij}\,w^{\mathrm{top}}_{ij}\,\ell_{ij}}
     {\sum_{(i,j)\in\mathcal{S}}
       w^{\mathrm{size}}_{ij}\,w^{\mathrm{top}}_{ij}},
\]
where $\mathcal{S}$ is the set of sibling pairs $(i,j)$ with a shared
prefix and $y_i < y_j$. The total loss then combines the original
log-cardinality regression loss with the ranking loss:
\[
\mathcal{L} \;=\; \mathcal{L}_{\mathrm{card}} \;+\; \lambda\,\mathcal{L}_{\mathrm{rank}}.
\]
We use $\alpha = 1.0$, $\beta = 1.0$, and $\lambda = 3.0$ in our experiments below.
Keeping the cardinality regression loss is important for two reasons. First, removing it leads
to model collapse, where the model predicts the same cardinality for all siblings, since this 
leads to a very small ranking loss term. The cardinality regression loss prevents this by grounding the
predictions in real cardinalities. Secondly, accurate cardinality prediction is still important apart from
join ordering, e.g., as information for choosing the right join operators.

\subsection{P-error}
A common metric to determine the quality of plans found by different estimators is the \emph{p-error}~\cite{perror}.
Let $Q$ be a BGP of $n$ triples and $P$ be a plan (i.e., join ordering) for this plan. The \emph{$C_{\text{out}}$} cost~\cite{cout} recursively defines the cost of a join $v$ as
\begin{equation}
    C_{\text{out}}(v) = |\Omega(v)| + C_{\text{out}}(l) + C_{\text{out}}(r).
\end{equation}
The cost of leaf nodes (i.e., triple patterns) is defined as zero. Hence, this cost equals the sum of
all join output cardinalities. It is a simple yet good logical cost model, as it reflects the total amount of
results that are internally materialized during query execution.

The p-error now compares the cost of a plan $\widehat{P}$ guided by cardinality estimates $\widehat{C}$, evaluated by true cardinalities $C$, i.e. $C_{\text{out}}(\widehat{P}, C)$ with the cost of the optimal plan $P^*$, evaluated by true cardinalities, i.e. $C_{\text{out}}(P^*, C)$. It is given as
\begin{equation}
    \text{p-error}(Q) = \frac{C_{\text{out}}(\widehat{P}, C)}{C_{\text{out}}(P^*, C)}.
\end{equation}
It hence evaluates how many times worse the actual plan found using the estimator is from the optimal plan. A p-error of 1 means the optimizer found the optimal plan, while, e.g., a p-error of 2 means the plan is twice as bad as the optimal plan.
\begin{figure}[t!]
\centering
\includegraphics[width=1\textwidth]{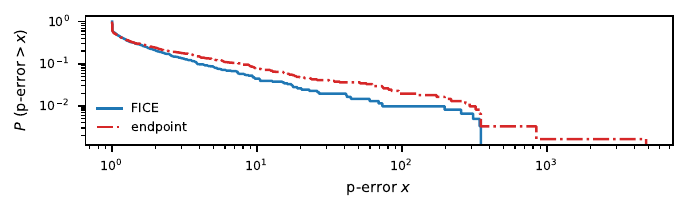}
\caption{Complementary CDF of p-error for FICE and Qlever's optimizer
on held-out DBpedia100k queries.}
\label{fig:ranking-ccdf}
\end{figure}

\subsection{Experimental Setup}
To perform the evaluation, we generated the above-described sibling datasets for all 10 KGs, where we used up to 3000 queries per KG and shape, and for each of those queries, we generated a single linear plan, where at each level, we generated a maximum of 6 siblings. We then train FICE with the combined loss on
nine training KGs from the main experiments, and evaluated on the
held-out DBpedia100k. For each
held-out BGP we ran a $C_{\text{out}}$ DP~\cite{DP,dpccp} enumerator over all connected
sub-BGPs, considering any join-tree shape, bushy or left-deep,
using FICE's predictions, and compared the resulting
plan against the plan emitted by Qlever's native optimizer for the
same BGP. Note that Qlever itself uses bushy enumeration of join trees,
so we are comparing both estimators in the same plan space.
Ground-truth intermediate cardinalities are obtained from
the endpoint directly by execution. Hence, for the $c$-out comparison, both planners' chosen
plans are evaluated under true intermediate sizes.

\subsection{Results}

Table~\ref{tab:ranking-results} summarises the main results. FICE
improves over the endpoint planner on geometric mean p-error and has large
gaps to it in the heavy tail: a $2\times$ smaller p99 ($195.01$ vs.\ $342.83$) and
a $14\times$ smaller worst case ($349.7$ vs.\ $4{,}822$). This shows that
on average, FICE finds a plan closer to the optimal plan significantly more often,
and importantly, it prevents catastrophic failure modes of the internal optimizer (i.e., the p90/p99 tails). This can also be clearly seen in figure~\ref{fig:ranking-ccdf}, which, analogously to the results in section~\ref{subsec:card-est}, shows the CCDF of the p-error.
While both FICE and the internal optimizer are close in the typical range, the tail behavior
of FICE is significantly better. The average ratio of FICE's $C_{\text{out}}$ cost and that of the
native optimizer is 0.79, i.e., the plans of FICE are on average $21\%$ better in terms of cost.

In summary, the results show that with the only addition of the ranking loss, while keeping the rest
of the FICE architecture identical, the approach can enhance the cost of plans compared to a production-grade triplestore optimizer. Those are only first results in the direction of join order optimization, and
the next steps are thus to integrate FICE into triplestores and evaluate how far it can enhance the
overall runtime of query workloads.

\begin{figure}[h!]
\includegraphics[width=\textwidth]{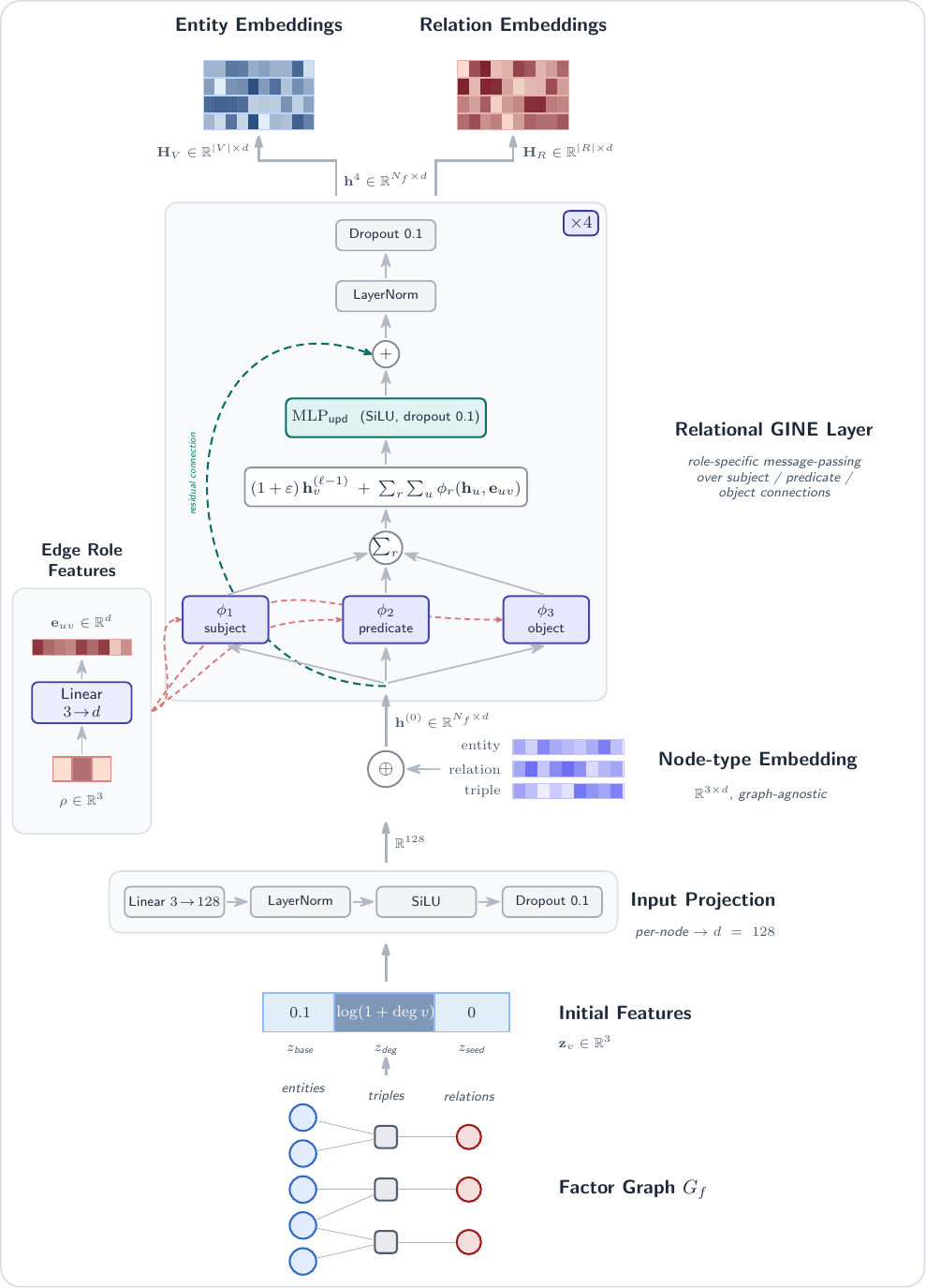}
\caption{Detailed Illustration of the encoder GNN, showing the transformation of nodes in the factor graph into contextualized embeddings for entities and relations.} \label{fig:encoder}
\end{figure}

\begin{figure}[h!]
\centering
\includegraphics[width=0.7\textwidth]{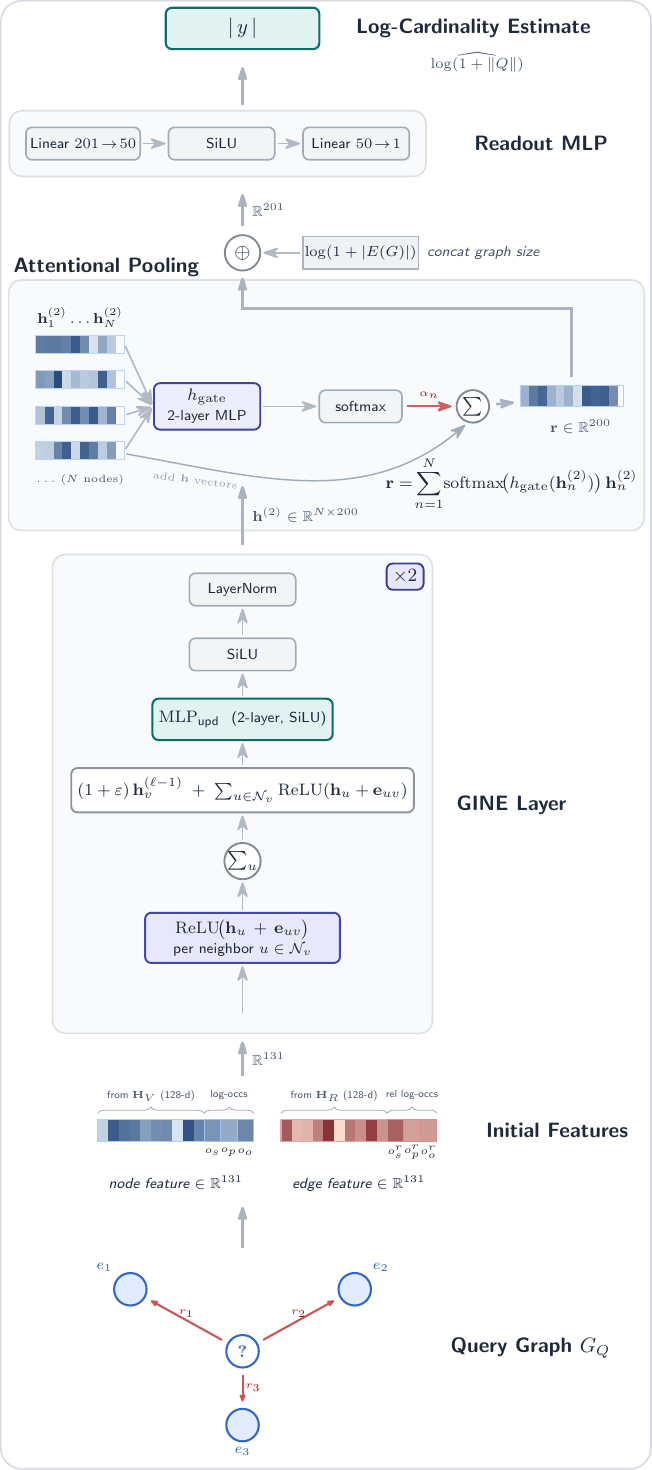}
\caption{Detailed Illustration of the decoder GNN, showing how a query graph gets enhanced with embeddings from $f^E$ and transformed into a scalar cardinality estimate.} \label{fig:decoder}
\end{figure}
\fi
\end{document}